\documentclass[pra,showpacs,preprintnumbers,amsmath,amssymb,superscriptaddress]{revtex4}

\usepackage{bm}
\usepackage{graphicx}
\usepackage{amsbsy}
\usepackage{amsmath}
\usepackage{amsfonts}
\usepackage{amsthm}

\begin{document}

\theoremstyle{plain}
\newtheorem{theorem}{Theorem}
\newtheorem{lemma}[theorem]{Lemma}
\newtheorem{corollary}[theorem]{Corollary}
\newtheorem{conjecture}[theorem]{Conjecture}
\newtheorem{proposition}[theorem]{Proposition}
\newcommand{\PT}{\mathrm{PTL}}
 \newcommand{\C}{\mathbb{C}}
  \newcommand{\F}{\mathbb{F}}
  \newcommand{\N}{\mathbb{N}}
  \renewcommand{\P}{\mathbb{P}}
  \newcommand{\R}{\mathbb{R}}
  \newcommand{\Z}{\mathbf{Z}}
  \renewcommand{\a}{\mathbf{a}}
  \renewcommand{\b}{\mathbf{b}}
  \renewcommand{\i}{\mathbf{i}}
  \renewcommand{\j}{\mathbf{j}}
  \renewcommand{\c}{\mathbf{c}}
  \newcommand{\e}{\mathbf{e}}
  \newcommand{\f}{\mathbf{f}}
  \newcommand{\g}{\mathbf{g}}
  \newcommand{\gl}{\mathbf{GL}}
  \newcommand{\m}{\mathbf{m}}
  \newcommand{\n}{\mathbf{n}}
  \newcommand{\bNP}{\mathbf{NP}}
  \newcommand{\bNPC}{\mathbf{NPC}}
  \newcommand{\p}{\mathbf{p}}
  \newcommand{\bP}{\mathbb{P}}
  \newcommand{\bPo}{\mathbf{Po}}
  \newcommand{\q}{\mathbf{q}}
  \newcommand{\s}{\mathbf{s}}
  \newcommand{\bt}{\mathbf{t}}
  \newcommand{\T}{\mathbf{T}}
  \newcommand{\U}{\mathbf{U}}
  \renewcommand{\u}{\mathbf{u}}
  \renewcommand{\v}{\mathbf{v}}
  \newcommand{\V}{\mathbf{V}}
  \newcommand{\w}{\mathbf{w}}
  \newcommand{\W}{\mathbf{W}}
  \newcommand{\x}{\mathbf{x}}
  \newcommand{\X}{\mathbf{X}}
  \newcommand{\y}{\mathbf{y}}
  \newcommand{\Y}{\mathbf{Y}}
  \newcommand{\z}{\mathbf{z}}
  \newcommand{\0}{\mathbf{0}}
  \newcommand{\1}{\mathbf{1}}
  \newcommand{\Gam}{\mathbf{\Gamma}}
  \newcommand{\bGamma}{\Gam}
  \newcommand{\Lam}{\mathbf{\Lambda}}
  \newcommand{\lam}{\mbox{\boldmath{$\lambda$}}}
  \newcommand{\bA}{\mathbf{A}}
  \newcommand{\bB}{\mathbf{B}}
  \newcommand{\bC}{\mathbf{C}}
  \newcommand{\bH}{\mathbf{H}}
  \newcommand{\bL}{\mathbf{L}}
  \newcommand{\bM}{\mathbf{M}}
  \newcommand{\bc}{\mathbf{c}}
  \newcommand{\cA}{\mathcal{A}}
  \newcommand{\cB}{\mathcal{B}}
  \newcommand{\cC}{\mathcal{C}}
  \newcommand{\cD}{\mathcal{D}}
  \newcommand{\cE}{\mathcal{E}}
  \newcommand{\cF}{\mathcal{F}}
  \newcommand{\cG}{\mathcal{G}}
  \newcommand{\cH}{\mathcal{H}}
  \newcommand{\cI}{\mathcal{I}}
  \newcommand{\cL}{\mathcal{L}}
  \newcommand{\cM}{\mathcal{M}}
  \newcommand{\cO}{\mathcal{O}}
  \newcommand{\cP}{\mathcal{P}}
  \newcommand{\cR}{\mathcal{R}}
  \newcommand{\cS}{\mathcal{S}}
  \newcommand{\cT}{\mathcal{T}}
  \newcommand{\cU}{\mathcal{U}}
  \newcommand{\cV}{\mathcal{V}}
  \newcommand{\cW}{\mathcal{W}}
  \newcommand{\cX}{\mathcal{X}}
  \newcommand{\cY}{\mathcal{Y}}
  \newcommand{\cZ}{\mathcal{Z}}
  \newcommand{\rE}{\mathrm{E}}
  \newcommand{\rH}{\mathrm{H}}
  \newcommand{\rU}{\mathrm{U}}
  \newcommand{\Cp}{\mathrm{Cap\;}}
  \newcommand{\lan}{\langle}
  \newcommand{\ran}{\rangle}
  \newcommand{\an}[1]{\lan#1\ran}
  \def\diag{\mathop{{\rm diag}}\nolimits}
  \newcommand{\hs}{\hspace*{\parindent}}
  \newcommand{\cl}{\mathop{\mathrm{Cl}}\nolimits}
  \newcommand{\tr}{\mathop{\mathrm{Tr}}\nolimits}
  \newcommand{\Aut}{\mathop{\mathrm{Aut}}\nolimits}
  \newcommand{\argmax}{\mathop{\mathrm{arg\,max}}}
  \newcommand{\Eig}{\mathop{\mathrm{Eig}}\nolimits}
  \newcommand{\Gr}{\mathop{\mathrm{Gr}}\nolimits}
  \newcommand{\Fr}{\mathop{\mathrm{Fr}}\nolimits}
  \newcommand{\trans}{^\top}
  \newcommand{\opt}{\mathop{\mathrm{opt}}\nolimits}
  \newcommand{\per}{\mathop{\mathrm{perm}}\nolimits}
  \newcommand{\haff}{\mathrm{haf\;}}
  \newcommand{\perio}{\mathrm{per}}
  \newcommand{\conv}{\mathrm{conv\;}}
  \newcommand{\Cov}{\mathrm{Cov}}
  \newcommand{\inter}{\mathrm{int}}
  \newcommand{\dist}{\mathrm{dist}}
  \newcommand{\inn}{\mathrm{in}}
  \newcommand{\grank}{\mathrm{grank}}
  \newcommand{\mrank}{\mathrm{mrank}}
  \newcommand{\krank}{\mathrm{krank}}
  \newcommand{\out}{\mathrm{out}}
  \newcommand{\orient}{\mathrm{orient}}
  \newcommand{\Pu}{\mathrm{Pu}}
  \newcommand{\rdc}{\mathrm{rdc}}
  \newcommand{\range}{\mathrm{range\;}}
  \newcommand{\Sing}{\mathrm{Sing\;}}
  \newcommand{\topo}{\mathrm{top}}
  \newcommand{\undir}{\mathrm{undir}}
  \newcommand{\Var}{\mathrm{Var}}
  \newcommand{\rC}{\mathrm{C}}
  \newcommand{\rF}{\mathrm{F}}
  \newcommand{\rL}{\mathrm{L}}
  \newcommand{\rM}{\mathrm{M}}
  \newcommand{\rO}{\mathrm{O}}
  \newcommand{\rR}{\mathrm{R}}
  \newcommand{\rS}{\mathrm{S}}
  \newcommand{\rT}{\mathrm{T}}
  \newcommand{\pr}{\mathrm{pr}}
  \newcommand{\inte}{\mathrm{int}}
  \newcommand{\inv}{\mathrm{inv}}
  \newcommand{\pers}{\per_s}
  \newcommand{\del}{\boldsymbol{\delta}}
  \renewcommand{\alph}{\boldsymbol{\alpha}}
  \newcommand{\bet}{\boldsymbol{\beta}}
  \newcommand{\gam}{\boldsymbol{\gamma}}
  \newcommand{\sig}{\boldsymbol{\sigma}}
  \newcommand{\zet}{\boldsymbol{\zeta}}
  \newcommand{\et}{\boldsymbol{\eta}}
  \newcommand{\xit}{\boldsymbol{\xi}}
  \newcommand{\perm}{\mathrm{perm\;}}
  \newcommand{\adj}{\mathrm{adj\;}}
  \newcommand{\rank}{\mathrm{rank\;}}
  \newcommand{\set}[1]{\{#1\}}
  \newcommand{\spec}{\mathrm{spec\;}}
  \newcommand{\supp}{\mathrm{supp\;}}
  \newcommand{\Tr}{\mathrm{Tr\;}}
  \newcommand{\vol}{\text{vol}}
  
  \def\be{\begin{eqnarray}}
\def\ee{\end{eqnarray}}
\def\bee{\begin{eqnarray*}}
\def\eee{\end{eqnarray*}}
 \def\pmx{\begin{pmatrix}}
 \def\emx{\end{pmatrix}}
 \def\bsq{\begin{subequations}}
\def\esq{\end{subequations}}
 \def\zero{ \begin{smallmatrix}  0 & 0 \\ 0 & 0 \end{smallmatrix} }
 \def\zerocol{  \begin{smallmatrix} 0 \\ 0 \end{smallmatrix} }
  \def\zerorow{  \begin{smallmatrix} 0 & 0 \end{smallmatrix} }
\def\e{\epsilon}
 \def\iff{\Leftrightarrow}
   \def\imp{\Rightarrow}
        \def\tr{\hbox{\rm Tr} \, }
          \def\ran{\hbox{\rm Range ~}}
          \def\trp{\hbox{\rm Tr} }
     
     \def\half{{\textstyle \frac{1}{2}}}
     \def\nn{\nonumber}
     \def\rank{{\rm rank}}
          \def\range{{\rm range}~}
\def\cB{{\cal B}}
\def\cE{{\cal E}}

\def\cH{{\cal H}}
\def\cD{{\cal D}}
\def\cV{{\cal V}}
\def\cB{{\cal B}}
\def\cP{{\cal P}}
\def\cF{{\cal F}}
\def\cG{{\cal G}}
\def\cP{{\cal P}}

\def\eps{\epsilon}

 \def\halfs{{ \frac{1}{2}}}
\def\ds{\displaystyle}
 \def\ts{\textstyle}
\def\bra{\langle}
\def\ket{\rangle}
\def\kb{ \ket \bra }
\def\trho{ \tilde{\rho} }
\def\rt2{ \frac{1}{\sqrt{2}} }
\def\lraw{\leftrightarrow}
\def\raw{\rightarrow}
\def\uparw{\uparrow}
\def\dwnarw{\downarrow}
\def\wh{\widehat}
           \def\wtd{\widetilde}

\theoremstyle{definition}
\newtheorem{definition}{Definition}

\theoremstyle{remark}
\newtheorem*{remark}{Remark}
\newtheorem{example}{Example}

\title{The minimum entropy output of a quantum channel is locally additive}

\author{Gilad Gour}\email{gour@ucalgary.ca}
\affiliation{Institute for Quantum Information Science and
Department of Mathematics and Statistics,
University of Calgary, 2500 University Drive NW,
Calgary, Alberta, Canada T2N 1N4}

\author{Shmuel Friedland}\email{friedlan@uic.edu}
\affiliation{Department of Mathematics, Statistics and Computer Science
University of Illinois at Chicago,
851 S. Morgan Street,
Chicago, IL 60607-7045}

\begin{abstract}
We show that the minimum von-Neumann entropy output of a quantum channel is locally additive. 
Hasting's counterexample for the additivity conjecture, makes this result quite surprising.
In particular, it indicates that the non-additivity of the minimum entropy output is a \emph{global} effect
of quantum channels.
\end{abstract}

\maketitle

\section{Introduction}

One of the most fundamental questions in quantum information concerns with the amount of information that can be transmitted reliably through a quantum channel. Despite of the significant progress in recent 
years~\cite{Holevo06,KR01,Kin02,Kin03,FHMV,Sho02,HHH09,AHW00,Shor04,AB04,Has09,Yard08,Shor04,Brandao,Fuk10,Bra11}, as pointed out in~\cite{Bra11},
this question remained surprisingly wide open. The main reason for that is related to
the additivity nature of the classical or quantum capacities of quantum channels to transmit information~\cite{Holevo06}. Recently, it was shown that both the Holevo expression for the classical capacity~\cite{Has09} and the quantum capacity~\cite{Yard08} are not additive in general. The additivity of the Holevo expression for the classical capacity was an open problem for more than a decade and was shown by Shor~\cite{Shor04} to be equivalent to three other additivity conjectures; namely, the additivity of entanglement of formation, the strong super-additivity of entanglement of formation, and the additivity of the minimum entropy output of a quantum channel. 

In~\cite{Has09} Hastings gave a counterexample to the last of
the above additivity conjectures and thereby proved that they are all false. Hastings counterexamples (see also~\cite{Brandao}) exist in very high dimensions and an estimate of these extremely high dimensions can be found in~\cite{Fuk10}. Earlier, in~\cite{Shor04}, Shor pointed out that if the additivity conjectures were true, perhaps the first step towards proving them would be to prove local additivity. We show here that this local additivity conjecture is indeed true, despite the existence of counterexamples to the original additivity conjectures. Our results therefore demonstrate that the counterexamples to the original additivity conjecture exhibit a global effect of quantum channels.
  
As we pointed out in Appendix B of~\cite{FGA},
both the local and global additivity conjectures are false over the real numbers.
This in turn implies that a straightforward argument involving just directional derivatives
could not provide a proof of local additivity in the general complex case. Hence,
to show local additivity we use strongly the complex structure.

In quantum information theory, quantum channels are the natural generalizations of stochastic communication channels in classical information theory. They are described in terms of completely-positive trace preserving linear maps (CPT maps).
A CPT map $\mathcal{N}: H_{d_{\rm in}}\to H_{d_{\rm out}}$ takes the set of $d_{\rm in}\times d_{\rm in}$ Hermitian matrices 
$H_{d_{\rm in}}$ to a subset of the set of all $d_{\rm out}\times d_{\rm out}$ Hermitian matrices
$H_{d_{\rm out}}$. Any finite dimensional quantum channel can be characterized in terms of a unitary embedding followed by a partial trace (the Stinespring dilation theorem): for any CPT map $\mathcal{N}$ there exists an ancillary space of Hermitian matrices $H_{E}$ such that
$$
\mathcal{N}(\rho)=\tr_{E}\left[U(\rho\otimes |0\rangle_{E}\langle 0|) U^{\dag}\right]
$$
where $\rho\in H_{d_{\rm in}}$ and $U$ is a unitary matrix mapping states $|\psi\rangle|0\rangle_E$ with
$|\psi\rangle\in H_{d_{\text{in}}}$ to 
$H_{d_{\text{out}}}\otimes H_{E}$.
 
The minimum entropy output of a quantum channel $\mathcal{N}$ is defined by
$$
S_{\min}(\mathcal{N})\equiv\min_{\rho\in H_{d_{\text{in}},+,1}}S\left(\mathcal{N}(\rho)\right)\;,
$$ 
where $H_{d_{\text{in}},+,1}\subset H_{d_{\text{in}}}$ is the set of all $d_{\rm in}\times d_{\rm in}$ positive semi-definite matrices with trace one (i.e. density matrices), and $S(\rho)=-\tr(\rho\log\rho)$ is the von-Neumann entropy. 
Since the von-Neumann entropy is concave it follows that the minimization can be taken over all rank one matrices
$\rho=|\psi\rangle\langle\psi|$ in $H_{d_{\text{in}},+,1}$. 

For any such rank one density matrix $\rho$ we can define a bipartite pure state
$|\Psi\rangle=U|\psi\rangle|0\rangle_{E}$ in the bipartite subspace $\mathcal{K}\equiv \{|\Psi\rangle\big|\;|\psi\rangle\in H_{d_{\rm in}}\}$. We therefore find that the minimum
entropy output of the channel $\mathcal{N}$ can be expressed in terms of the entanglement of the bipartite subspace 
$\mathcal{K}$ defined by 
$$
E(\mathcal{K})\equiv\min_{|\phi\rangle\in\mathcal{K}\;,\;\|\phi\|=1}E(|\phi\rangle)\;,
$$
where $E(|\phi\rangle)\equiv S\left(\tr_{E}(|\phi\rangle\langle\phi|)\right)$ is the entropy of entanglement.
In~\cite{GN} it was pointed out that $E(\mathcal{K})=0$ unless $\dim\mathcal{K}\leq (d_{\rm out}-1)(\dim H_{E}-1)$.
This claim follows directly from the fact that the number of (bipartite) states in an unextendible product basis is 
at least $d_{\rm out}+\dim H_{E}-1$~\cite{Ben99}.

With these notations, the non-additivity of the minimum entropy output of a quantum channel is equivalent to 
the existence of two subspaces $\mathcal{K}_1\subset\mathbb{C}^{n_1}\otimes\mathbb{C}^{m_1}$ and $\mathcal{K}_2\subset\mathbb{C}^{n_2}\otimes\mathbb{C}^{m_2}$ such that
$$
E(\mathcal{K}_1\otimes\mathcal{K}_1)<E(\mathcal{K}_1)+E(\mathcal{K}_2)\;.
$$
In what follows we will prove the local additivity of entanglement of subspaces, which is equivalent to the local additivity of the minimum entropy output.  

The rest of this paper is organized as follows. In section~\ref{local} we find and simplify the first and second directional derivatives of the von-Neumann entropy of entanglement. In section~\ref{additive} we prove our main result of local additivity which is stated in Theorem~\ref{main} for the non-singular case. In section~\ref{sing} we prove Theorem~\ref{main} for the singular case. We end with a discussion in section~\ref{conc}.

\section{Local Minimum}\label{local}

Let $\mathcal{K}\subset\mathbb{C}^{n}\otimes\mathbb{C}^{m}$ be a subspace of bipartite entangled states.
Since the bipartite Hilbert space $\mathbb{C}^{n}\otimes\mathbb{C}^{m}$ is isomorphic to the Hilbert space
of all $n\times m$ complex matrices $\mathbb{C}^{n\times m}$, we can view any bipartite state 
$|\psi\rangle^{AB}=\sum_{i,j}x_{ij}|i\rangle|j\rangle$ in $\mathcal{K}$ as an $n\times m$ matrix $x$.
The reduced density matrix of $|\psi\rangle^{AB}$ is then given by $\rho_r\equiv\tr_{B}|\psi\rangle^{AB}\langle\psi|=xx^*$,
and the entropy of entanglement of $|\psi\rangle^{AB}$ is given by
\begin{equation}\label{entropy}
E(x)\equiv-\tr\left(xx^*\log xx^*\right)\;.
\end{equation}
In our notations, instead of using a dagger, we use $x^*$ to denote the hermitian conjugate of the matrix $x$.

If $x\in\mathcal{K}$ is a local minimum of $E$ in $\mathcal{K}$, then there exists a neighbourhood of $x$ in $\mathcal{K}$
such that $x$ is the minimum in that neighbourhood. Any state in the neighbourhood of $x$ can be written as
$ax+by$, where $a,b\in\mathbb{C}$ and $y\in\mathcal{K}$ is an orthogonal matrix to $x$; i.e. $\tr(xy^*)=0$.
We also assume that the state is normalized so that $|a|^2+|b|^2=1$.
Now, since the function $E(x)$ is independent on global phase, we can assume that $a$ is a positive \emph{real} number.
We can also assume that $b$ is real since we can absorb its phase into $y$ (adding a phase to $y$ will not change its orthogonality to $x$). Thus, any normalized state in the neighbourhood of $x$ can be written as
$$
\frac{x+ty}{\sqrt{1+t^2}}\;\;\text{with}\;\;\tr(xy^*)=0\;,
$$
where $t\equiv b/a$ is a small real number and $y$ is normalized (i.e. $\tr(yy^*)=1$). 

\begin{definition}\label{maindef}
$\;$\\
\textbf{(a)} A matrix $x\in\mathcal{K}$ is said to be a critical point of $E(x)$ in $\mathcal{K}$ if
$$
D_{y}E(x)\equiv \frac{d}{dt}E\left(\frac{x+ty}{\sqrt{1+t^2}}\right)\Big|_{t=0}=0\;\;\;\forall\;y\in x^{\perp}
$$
where the notation $D_{y}E(x)$ indicate that we are taking the directional derivative of $E$ in the direction of $y$,
and $x^\perp\subset\mathcal{K}$ denotes the subspace of all the matrices $y$ in $\mathcal{K}$ for which
$\tr(xy^*)=0$.\\
\textbf{(b)}  A matrix $x\in\mathcal{K}$ is said to be a non-degenerate local minimum of $E(x)$ in $\mathcal{K}$ if it is critical
and
$$
D_{y}^{2}E(x)\equiv \frac{d^2}{dt^2}E\left(\frac{x+ty}{\sqrt{1+t^2}}\right)\Big|_{t=0}>0\;\;\;\forall\;y\in x^{\perp},
$$ 
were we also allow $D_{y}^{2}E(x)=+\infty$.
Moreover, a critical $x\in\mathcal{K}$ is said to be degenerate if there exists at least one direction $y$ such that
$D_{y}^{2}E(x)=0$.
\end{definition} 

In order to prove local additivity we will need to calculate the above directional derivatives.
This can be done by expressing the logarithm as an integral~\cite{Yard} (see also~\cite{OP04,Petz}). However, in this technique 
all the quantities are expressed by integrals, and some of these integral expressions do not lead to additivity
in a transparent way, as the divided difference method does.
We therefore apply below a new technique that is based on the \emph{divided difference}~\cite[(6.1.17)]{HJ99}.
One of the advantages of the divided difference approach, is that it enables one to calculate and express 
all directional derivatives \emph{explicitly} with no integrals involved. Before introducing the divided difference approach, we will first discuss briefly the affine parametrization.

In our calculations we will assume that $x$ is diagonal (or equivalently, the bipartite state $x$ represents is given in its 
Schmidt form). This assumption follows from
the singular value decomposition theorem; namely, we can always find unitary matrices $u\in\mathbb{C}^{n\times n}$ and $v\in\mathbb{C}^{m\times m}$ such that $uxv$ is an $n\times m$ diagonal matrix with non-negative real numbers (the singular values of $x$) on the diagonal. Since $E(x)=E(uxv)$ we can assume without loss of generality
that $x$ is a diagonal matrix.  

\subsection{The Affine Parametrization}

Up to second order in $t$ we have
\begin{align}
\rho(t) \equiv\frac{(x+ty)(x^{*}+ty^{*})}{1+t^2}
&=\left(xx^{*}+t(xy^{*}+yx^{*})+t^{2}yy^{*}\right)(1-t^{2})\nonumber\\
&=xx^{*}+t(xy^{*}+yx^{*})+t^2(yy^{*}-xx^{*})
=\rho+t\gamma_0+t^2\gamma_1\;,
\end{align}
where $\rho=xx^*$, $\gamma_0\equiv xy^*+yx^*$, and $\gamma_1\equiv yy^*-xx^*$.
Note that $\tr\rho=1$ and $\tr\gamma_0=\tr\gamma_1=0$, 
where without loss of generality we assumed $\tr (yy^*)=1$ since we can absorb the normalization factor of $y$ into $t$. We are interested in taking the first and second derivative of 
$$
E\left(\frac{x+ty}{\sqrt{1+t^2}}\right)=S(\rho(t))=S(\rho+t\gamma_0+t^2\gamma_1)\;.
$$

In this section we assume that $\rho=xx^*$ is an $n\times n$ non-singular matrix.
Denote 
$$
\sigma(t)\equiv\rho+t\gamma_0\;.
$$
In the next proposition we  relate $S(\rho(t))$ with $S(\sigma(t))$. 
\begin{proposition}\label{affparfor}
Let $\rho(t),\sigma(t),\rho,\gamma_0$ and $\gamma_1$ as above. Then
\begin{equation}\label{affparfor1}
S(\rho(t))=S(\sigma(t))-t^2\tr\left[\gamma_1 \log\rho\right]
+O(t^3)
\end{equation}
\end{proposition}

\begin{proof}
Since $\rho$ is non-singular, also $\rho(t)$ and $\sigma(t)$ are non-singular for small enough $t$. Thus, $I-\rho(t)<I$ for small $t$. Using the Taylor expansion 
$$
\log\rho(t)=\log[I-(I-\rho(t))]=-\sum_{n=1}^{\infty}\frac{\left(I-\sigma(t)-t^2\gamma_1\right)^n}{n}\;, 
$$
we get
$$
-\tr\left[\rho\log\rho(t)\right]=\sum_{n=1}^{\infty}\frac{1}{n}\tr\left[\rho\left(I-\sigma(t)-t^2\gamma_1\right)^n\right]\;.
$$
Expanding the term in the trace above up to second order in $t$ gives
$$
\tr\left[\rho\left(I-\sigma(t)-t^2\gamma_1\right)^n\right]=\tr\left[\rho\left(I-\sigma(t)\right)^n\right]
+t^2n\tr\left[\rho(I-\rho)^{n-1}\gamma_1\right]+O(t^3)\;.
$$
We therefore have
$$
-\tr\left[\rho\log\rho(t)\right]=-\tr\left[\rho\log\sigma(t)\right]+t^2\sum_{n=1}^{\infty}\tr\left[\rho(I-\rho)^{n-1}\gamma_1\right]+O(t^3)\;.
$$
Since $\rho^{-1}=\sum_{n=1}^{\infty}(I-\rho)^{n-1}$ and $\tr(\gamma_1)=0$ we conclude
$$
\tr\left[\rho\log\rho(t)\right]=\tr\left[\rho\log\sigma(t)\right]+O(t^3)\;.
$$
Thus, $$\tr[\rho(t)\log\rho(t)]=\tr\left[\sigma(t)\log\sigma(t)\right]+t^2\tr\left[\gamma_1 \log\rho\right]+O(t^3).$$
This completes the proof.
\end{proof}

This simple relation between $S(\rho(t))$ and $S(\sigma(t))$ is very useful since now we can focus on the Taylor expansion of the simpler function $S(\sigma(t))$. 

 \subsection{The method of divided difference}
 
To calculate the first and second derivatives of $S(\sigma(t)$, we first evaluate the Taylor expansion of 
a complex valued function $f:\mathbb{C}\to\mathbb{C}$, which we later assume can be extended to act on 
$n\times n$ complex matrices. 
 
 We will make use of the notion of the \emph{divided difference} for $f$, which we refer the reader to~\cite[(6.1.17)]{HJ99}
 for more details. The divided difference for a function $f:\C\to\C$,
 given a sequence of distinct complex points, 
 $\alpha_i\in\C, i=1,\ldots,n$, is defined for $i=0,1$ by
 \begin{align}\label{defdivdif1}
 & \triangle^0 f(\alpha_1):=f(\alpha_1)\\
 & \triangle^1 f(\alpha_1,\alpha_2)\equiv\triangle f(\alpha_1,\alpha_2):=\frac{f(\alpha_1)-f(\alpha_2)}{\alpha_1-\alpha_2},
 \end{align}
 and defined inductively by
 \be
 \triangle^i f(\alpha_1,\ldots,\alpha_i,\alpha_{i+1})=
\frac{\triangle^{i-1} f(\alpha_1,\ldots,\alpha_{i-1},\alpha_{i}) - \triangle^{i-1} f(\alpha_1,\ldots,\alpha_{i-1},\alpha_{i+1})}
 {\alpha_i-\alpha_{i+1}},
 \ee
 for $i=2,3,\ldots,n$.  It is well known that $\triangle^i f(\alpha_1,\ldots,\alpha_i,\alpha_{i+1})$ is a symmetric function in $\alpha_1,\ldots,\alpha_{i+1}$,
 e.g. \cite[p'393]{HJ99}.  For points that are not distinct it is defined by an appropriate limit. For example, for $x\ne y$ we have
 \begin{align}\label{g1g}
 & \triangle f(x,x)=f'(x)\nonumber\\ 
 & \triangle^2 f(x,x,y)=\frac{f'(x)}{(x-y)} -\frac{f(x)-f(y)}{(x-y)^2}\\ 
 & \triangle^2 f(x,x,x)=\frac{1}{2} f''(x).\label{ququ}
 \end{align}
 Note that~\eqref{ququ} can be obtained from~\eqref{g1g} by setting $h\equiv y-x\to 0$ and expending
 $f(y)=f(x+h)=f(x)+hf'(x)+\frac{1}{2}h^2f''(x)+O(h^3)$.

 \begin{theorem}\label{taylor}  Let $A=\diag(\alpha_1,\ldots,\alpha_n)\in\C^{n\times n}$ be a diagonal square matrix, and
 $B=[b_{ij}]\in\C^{n\times n}$ be a complex square matrix. Assume that $f(x):\C\to\C$ satisfy one of the following conditions:
 \begin{enumerate}
 \item\label{smoothcase1}  $f(x)$ is an analytic function in some domain $\cD\subset \C$ which contains $\alpha_1,\ldots,\alpha_n$,
 and can be approximated uniformly in $\cD$ by polynomials.
 \item\label{smoothcase2}  $\alpha_1,\ldots,\alpha_n$ are in a real open interval $(a,b)$ and $f$ has two continuous derivatives in $(a,b)$.
 \end{enumerate}
 Then
 \begin{equation}\label{polcase1}
 f(A+tB)=f(A)+tL_A(B)+t^2Q_A(B)+O(t^3)
 \end{equation}
 Here $L_A:\C^{n\times n}\to\C^{n\times n}$ is a linear operator, and $Q_B:\C^{n\times n}\to \C^{n\times n}$ is a quadratic homogeneous noncommutative
 polynomial in $B$. For $i,j=1,\ldots,n$ we have
  \begin{align}\label{LABform}
& [L_A(B)]_{ij}=\triangle f(\alpha_i,\alpha_j)b_{ij}=\frac{f(\alpha_i)-f(\alpha_j)}{\alpha_i-\alpha_j} b_{ij}\\
 \label{QABform}
& [Q_A(B)]_{ij}=\sum_{k=1}^n \triangle^2f(\alpha_i,\alpha_k,\alpha_j) b_{ik}b_{kj}.
 \end{align}
 In particular
 \begin{align}
&  \tr(L_A(B))=\sum_{j=1}^{n}f'(\alpha_j)b_{jj}\\
 \label{tracQAB}
& \tr(Q_A(B))=\sum_{i,j=1}^n \frac{f'(\alpha_i)-f'(\alpha_j)}{2(\alpha_i-\alpha_j)}b_{ij}b_{ji}.
 \end{align}
\end{theorem}
\begin{remark}
The expansion above can be naturally generalized to higher than the second order, but for the purpose of this article, we will
only need to expand $f(A+tB)$ up to the second order in $t$. Moreover, for our purposes we will only need to assume that the $\alpha_i$ are real and the condition 2 on $f$ holds. We kept condition 1 on $f$ in the theorem just to be a bit more general. 
\end{remark}
Note that in all the expressions above, one must identify $\alpha_i=\alpha_j$ with the limit $\alpha_j\to\alpha_i$.
For example, the term
$$
\frac{f'(\alpha_i)-f'(\alpha_j)}{2(\alpha_i-\alpha_j)}=\frac{1}{2}f''(\alpha_i)\;\;\;\text{for}\;\;\alpha_i=\alpha_j\;.
$$
In particular, note that if $B$ is diagonal, Eq.~\eqref{tracQAB} gives the known second order term of the Taylor expansion.
 \proof  
 From the conditions on $f$, it is enough to prove the theorem assuming $f$ is a polynomial.
 By linearity, it is enough to prove all the claims for $f(x)=x^m$.
 Clearly, in the expension
 $$
 (A+tB)^{m}=A^{m}+tL_A(B)+t^2Q_A(B)+O(t^3)
 $$
 we must have 
 \begin{align}
 & L_A(B)=\sum_{0\le p,q,\; p+q=m-1} A^pBA^q,\label{121}\\ 
 & Q_A(B)=\sum_{0\le p,q,r,\; p+q+r=m-2} A^pBA^qBA^r,
 \label{LABQABfor}
 \end{align}
 where we expanded $(A+tB)^m$ up to first and second order in $t$. All that is left to show is that 
 these matrices coincide with the ones defined in~Eqs.~(\ref{LABform},\ref{QABform}).
 
 Indeed, since $A$ is diagonal, the matrix elements of the $L_{A}(B)$ in Eq.(\ref{121}) are given by
 $$
 [L_A(B)]_{ij}=\sum_{0\le p,q,\; p+q=m-1}\alpha_{i}^{p}\alpha_{j}^{q}b_{ij}=\frac{\alpha_{i}^{m}-\alpha_{j}^{m}}{\alpha_{i}-\alpha_{j}}b_{ij}\;,
 $$
 which is equal to the exact same matrix elements given in Eq.(\ref{LABform}).
 
In the same way, since $A$ is diagonal, observe that the matrix elements of the $Q_{A}(B)$ in Eq.(\ref{LABQABfor}) are given by
 $$
 [Q_A(B)]_{ij}=\sum_{k=1}^{n}\sum_{0\le p,q,r,\; p+q+r=m-2}\alpha_{i}^{p}\alpha_{k}^{q}\alpha_{j}^{r}b_{ik}b_{kj}\;.
 $$
 On the other hand, a straightforward calculation gives for $f(x)=x^m$
 $$
 \triangle^2 x^m(\alpha_i,\alpha_k,\alpha_j)=\sum_{0\le p,q,r,\; p+q+r=m-2} \alpha_{i}^{p} \alpha_{k}^{q} \alpha_{j}^{r}.
 $$
 Thus, the expressions in Eq.~(\ref{QABform}) and Eq.~(\ref{LABQABfor}) for $Q_{A}(B)$ are the same.

 We now prove Eq.~\eqref{tracQAB}.  Observe first that Eq.~\eqref{QABform} yields
 \begin{equation}\label{tracQAB1}
 \tr(Q_A(B))=\sum_{i,j=1}^n \triangle ^2 f(\alpha_i,\alpha_i, \alpha_j)b_{ij}b_{ji}\;,
 \end{equation}
 where we have used the symmetry $\triangle^2f(\alpha_i,\alpha_j, \alpha_i)=\triangle^2f(\alpha_i,\alpha_i, \alpha_j)$.
 Now, since $b_{ij}b_{ji}$ is symmetric under an exchange between $i$ and $j$, we can replace 
 $\triangle ^2 f(\alpha_i,\alpha_i, \alpha_j)$ in Eq.~(\ref{tracQAB1}) with 
 $$
 \frac{1}{2}\left[\triangle ^2 f(\alpha_i,\alpha_i, \alpha_j)+\triangle ^2 f(\alpha_j,\alpha_j, \alpha_i)\right]
 =\frac{1}{2}\triangle f'(\alpha_i,\alpha_j)\;,
 $$
 where for the last equality we used Eq.~\eqref{g1g}. This completes the proof.\qed

We now use the above theorem for the Taylor expansion of the function $S(\sigma(t))$ in the neighbourhood of $t=0$.

\subsection{The first and second derivatives of $E(x)$}

We first assume that $\rho$ is non singular. The case where $\rho$ is singular will be treated separately in section~\ref{sing}.
\begin{theorem}
Let $\rho=\diag\{p_1,\ldots,p_{n}\}$ with $p_j>0$
for $j=1,\ldots,n$. For this case, we get the following expressions:
\begin{align}
D_{y}^{1}E(x) & \equiv \frac{d}{dt}S(\rho(t))\Big|_{t=0}=-\tr(\gamma_0\log \rho)\nonumber\\
D_{y}^{2}E(x) & \equiv\frac{d^2}{dt^2}S(\rho(t))\Big|_{t=0}
=-2\left(\tr\left[\gamma_1 \log\rho\right]+\sum_{j,k} \frac{\log p_{j}-\log p_{k}}{2(p_j-p_k)}\left|(\gamma_0)_{jk}\right|^2\right)\;.
\label{se}
\end{align}
\end{theorem}

\begin{remark}
The condition for $x\in\mathcal{K}$ to be critical is $D_{y}^{1}E(x)=0$ which is equivalent to $\tr[(xy^*+yx^*)\log xx^*]=0$ for all $y\in\mathcal{K}$
such that $\tr(xy^*)=0$. Moreover, if $x$ is critical then we also have 
$D_{iy}^{1}E(x)=0$ for all $y\in x^\perp\subset\mathcal{K}$.
Hence, if $x$ is critical we must have 
\begin{equation}\label{critical}
\tr(xy^*\log xx^*)=0
\end{equation} 
for all $y\in x^\perp\subset\mathcal{K}$.   
\end{remark}

\begin{proof}
Theorem~\ref{taylor} implies that
 $$
S(\rho+t\gamma_0)=S(\rho)+tL_{\rho}(\gamma_0)+t^2Q_{\rho}(\gamma_0)+\mathcal{O}(t^3).
 $$
 where $L_\rho$ and $Q_\rho$ are the following linear and quadratic forms
\begin{align*}
 & L_{\rho}(\gamma)\equiv \sum_{i=1}^{n} g'(p_i)(\gamma_0)_{ii}\\
& Q_{\rho}(\gamma)\equiv \sum_{i=1}^{n}\sum_{j=1}^{n}
 \frac{g'(p_i)-g'(p_j)}{2(p_i-p_j)}(\gamma_0)_{ij}(\gamma_0)_{ji}\;,
 \end{align*}
and $g(t)\equiv-t\log t$. Note that the expressions for $L_{\rho}(\gamma_0)$ and $Q_{\rho}(\gamma_0)$ 
above are the traces of the analogous expressions given in theorem~\ref{taylor}, since
$S(\rho)$ is defined as the trace of the matrix $g(\rho)=-\rho\log \rho$.
 
Since $\gamma_0$ is hermitian with zero trace, and $g'(t)=-1-\log t$, we get
 \begin{eqnarray}
 L_{\rho}(\gamma_0) &=& -\tr(\gamma_0\log \rho)\nonumber\\
 Q_{\rho}(\gamma_0) &=& -\sum_{j,k} \frac{\log p_{j}-\log p_{k}}{2(p_j-p_k)}\left|(\gamma_0)_{j,k}\right|^2\;\label{qxy}\;.
\end{eqnarray}
Combining this with proposition~\ref{affparfor} proves the theorem.
\end{proof}

In the following lemma, we rewrite the expression in Eq.~(\ref{se}), which will be useful for the proof of local additivity.

\begin{lemma}\label{lemsecder}
Denote $w=(y+y^*)/2$, and $z=i(y-y^*)/2$. Denote also $r_{jk}=\sqrt{p_j/p_k}$, where $\{p_j\}_{i=1}^{n}$ are the eigenvalues of $\rho=xx^*$.
Then,
the expression in Eq.(\ref{se}) for $D_{y}^{2}E(x)$ can be rewritten as
\begin{align}
D_{y}^{2}E(x)  =-2E(x)
-\tr\left[(yy^*+y^*y) \log\rho\right]
 -2\sum_{j,k}\left(|w_{jk}|^2\Phi(r_{jk})+|z_{jk}|^2\Phi(-r_{jk})\right)\;,\label{gf}
\end{align}
where
\begin{equation}\label{Phi}
\Phi(r)\equiv\frac{1}{2}\frac{r+1}{r-1}\log r^{2}\;\;\;,\; r\in\mathbb{R},
\end{equation}
with the identification $\Phi(1)=2$.
\end{lemma}

\begin{proof}
The expression in Eq.(\ref{se}) for $D_{y}^{2}E(x)$ involves the terms $|(\gamma_0)_{jk}|^2$. The matrix 
$\gamma_0=xy^{*}+yx^{*}=xy^{*}+yx$ where $x=\diag\{\sqrt{p_1},\ldots,\sqrt{p_n}\}$. Note that $y^*=w+iz$ and $y=w-iz$, where $w$ and $z$ are the Hermitian matrices defined in the lemma. Thus,
$$
\gamma_0=xw+wx+i(xz-zx)\;.
$$
In terms of the matrix elements $w_{jk}$ and $z_{jk}$ of $w$ and $z$, we have
$$
(\gamma_0)_{jk}=(\sqrt{p_k}+\sqrt{p_j})w_{jk}+i(\sqrt{p_j}-\sqrt{p_k})z_{jk}\;.
$$
The square of this expression can be written as
\begin{align*}
|(\gamma_0)_{jk}|^2  =(\sqrt{p_j}+\sqrt{p_k})^2\left |w_{jk}\right|^2+(\sqrt{p_j}-\sqrt{p_k})^2\left|z_{jk}\right|^2
+i(p_j-p_k)(w_{jk}^{*}z_{jk}-w_{jk}z_{jk}^{*})
\end{align*}
Moreover, expressing back $w$ and $z$ interms of $y$ gives 
$i(w_{jk}^{*}z_{jk}-w_{jk}z_{jk}^{*})=(|y_{kj}|^2-|y_{jk}|^2)/2$.
We can therefore write
\begin{align*}
|(\gamma_0)_{jk}|^2 = (\sqrt{p_j}+\sqrt{p_k})^2\left |w_{jk}\right|^2+(\sqrt{p_j}-\sqrt{p_k})^2\left|z_{jk}\right|^2
+\frac{1}{2}(p_j-p_k)(|y_{kj}|^2-|y_{jk}|^2)\;.
\end{align*}
Substituting this expression, and the value for $\gamma_1=yy^*-xx^*$, into Eq.(\ref{se}) gives
$$
 -\frac{1}{2}D_{y}^{2}E(x) =E(x)+
 \tr\left[yy^* \log\rho\right]
+\sum_{j,k}\log \left(\frac{p_j}{p_k}\right)
\left\{\frac{(\sqrt{p_j}+\sqrt{p_k})^2}{2(p_j-p_k)}\left|w_{jk}\right|^2
+\frac{(\sqrt{p_j}-\sqrt{p_k})^2}{2(p_j-p_k)}\left|z_{jk}\right|^2+\frac{1}{4}(|y_{kj}|^2-|y_{jk}|^2)
\right\}\;.
$$
Note first that the term
$$
\frac{1}{4}\sum_{j,k}\log \left(\frac{p_j}{p_k}\right)
(|y_{kj}|^2-|y_{jk}|^2)
= \frac{1}{2}\tr\left[(y^*y-yy^*)\log\rho\right]\;.
$$
Moreover, denoting $r_{jk}=\sqrt{p_j/p_k}$ we get
$$
\frac{(\sqrt{p_j}+\sqrt{p_k})^2}{2(p_j-p_k)}\log \left(\frac{p_j}{p_k}\right)
=\frac{(r_{jk}+1)^{2}}{2(r_{jk}^2-1)}\log r_{jk}^{2}
=\frac{1}{2}\frac{r_{jk}+1}{r_{jk}-1}\log r_{jk}^{2}\equiv \Phi(r_{jk})
$$
Similarly,
$$
\frac{(\sqrt{p_j}+\sqrt{p_k})^2}{2(p_j-p_k)}\log \left(\frac{p_j}{p_k}\right)
=\frac{1}{2}\frac{r_{jk}-1}{r_{jk}+1}\log r_{jk}^{2}
= \Phi(-r_{jk})
$$
With these notations we get
$$
 -\frac{1}{2}D_{y}^{2}E(x)  =E(x)+
\frac{1}{2}\tr\left[(yy^*+y^*y) \log\rho\right]
+\sum_{j,k}\left(|w_{jk}|^2\Phi(r_{jk})+|z_{jk}|^2\Phi(-r_{jk})\right)
$$
This complete the proof.
\end{proof}

In the rest of the paper we will use the notations
\begin{align}
M_{x}(y)&\equiv\sum_{j,k=1}^{n}\left(|w_{jk}|^2\Phi(r_{jk})+|z_{jk}|^2\Phi(-r_{jk})\right)=\tr\left[w\Phi_{\rho}^{+}(w)+z\Phi_{\rho}^{-}(z)\right]
\nonumber\\
\Gamma_{x}(y)& \equiv -E(x)-\frac{1}{2}\tr\left[(y^*y+yy^*)\log xx^*\right]\;.\label{gam}
\end{align}
where $\Phi_{\rho}^{\pm}$ are self-adjoint linear operators defining in terms of the Hadamard product between the input matrix
and the matrix with elements $\Phi(\pm r_{jk})$. That is,
$$
\left[\Phi_{\rho}^{\pm}(w)\right]_{jk}=\Phi(\pm r_{jk})w_{jk}\;.
$$
With these notations we get that 
$D_{y}^{2}E(x)> 0$ if and only if
\begin{equation}\label{delta1}
M_{x}(y)< \Gamma_{x}(y)\;.
\end{equation}

\subsection{The complex structure and additional necessary condition}

If $D_{y}^{2}E(x)> 0$ for all $y$ orthogonal to $x$, then $D_{iy}^{2}E(x)$ is also positive
since $iy$ is orthogonal to $x$. That is,
\begin{equation}\label{delta2}
M_{x}(iy)< \Gamma_{x}(iy)=\Gamma_{x}(y)\;.
\end{equation}
Therefore, we get from Eqs.~(\ref{delta1},\ref{delta2}) that if $x$ is a non-degenerate local minimum then
\begin{align}\label{221}
\frac{1}{2}\left(M_{x}(y)+M_{x}(iy)\right)
= \sum_{j,k}|y_{jk}|^2\tilde{\Phi}(r_{jk})< \Gamma_{x}(y)\;,
\end{align}
where
\begin{equation}\label{Phi0}
\tilde{\Phi}(r):=\frac{1}{2}\left(\Phi(r)+\Phi(-r)\right)=\frac{1}{2}\frac{r^2+1}{r^2-1}\log r^2\;,
\end{equation}
with the identification $\tilde{\Phi}(\pm1)=1$. Let $\tilde{\Phi}_{\rho}$ be a self-adjoint linear operator defining in terms of the Hadamard product between the input matrix and the matrix with components $\tilde{\Phi}(r_{jk})$. With this notation the necessary condition given in Eq.(\ref{221}) can be written as
\be
\tr\left[y^*\tilde{\Phi}_{\rho}(y)\right]<\Gamma_{x}(y)\;.\label{222}
\ee

A simple analysis of the function $\tilde{\Phi}$ shows that $\tilde{\Phi}(r)\geq 1$ with equality if and only if $r=\pm 1$. 
Thus, Eq.(\ref{222}) also implies the following necessary condition on a local minimum:
$$
1 \leq \Gamma_{x}(y)\;,
$$
which can be written as
\begin{equation}\label{elegant}
E(y)-E(x)\geq 1-\frac{1}{2}\left[S(yy^*\|xx^*)+S(y^*y\|xx^*)\right]
\end{equation}
where 
$$
S(yy^*\|xx^*)\equiv \tr(yy^*\log yy^*)-\tr(yy^*\log xx^*)
$$
is the relative entropy. Since $S(yy^*\|xx^*)\geq 0$ with equality if and only if $yy^*=xx^*$, we always have
$S(yy^*\|xx^*)>0$ for $\tr(xy^*)=0$. Nevertheless, it is possible that $\tr(xy^*)=0$ and yet $S(yy^*\|xx^*)\leq 1$.
In such cases Eq.~\eqref{elegant} gives $E(y)\geq E(x)$ which is consistent with the fact that $x$ is a local min.
 
\section{Local Additivity}\label{additive}

We now state the main result of this paper.

\begin{theorem}\label{main}
Let $x^{(1)}$ and $x^{(2)}$ be two non-degenerate local minima of $E(x)$ in $\mathcal{K}^{(1)}\subset\C^{n_1\times m_1}$
and $\mathcal{K}^{(2)}\subset\C^{n_2\times m_2}$, respectively. Then, $x^{(1)}\otimes x^{(2)}$ is a non-degenerate local minimum
of $E(x)$ in $\mathcal{K}^{(1)}\otimes\mathcal{K}^{(2)}$. Moreover, if $x^{(1)}$ is degenerate local minimum and 
$x^{(2)}$ is non-degenerate local minimum, then $x^{(1)}\otimes x^{(2)}$ is a degenerate local minimum.
\end{theorem}

The theorem above implies, in particular, that if $x^{(1)}$ and $x^{(2)}$ are critical points of $E(x)$ in $\mathcal{K}^{(1)}$
and $\mathcal{K}^{(2)}$, respectively, then, $x^{(1)}\otimes x^{(2)}$ is a critical point of $E(x)$
in $\mathcal{K}^{(1)}\otimes\mathcal{K}^{(2)}$. This fact was observed in~\cite{AIM08} (see also~\cite{Maxim}), 
and later was stated in~\cite{FGA}.
It follows from the linearity in $y$ of the condition given in Eq.~(\ref{critical}) for critical points. 
More precisely, if $x^{(1)}$ and $x^{(2)}$ are critical points, then $x^{(1)}\otimes x^{(2)}$ is also critical if
(see Eq.~(\ref{critical}))
\begin{align*}
&0 =\tr\left[\left(x^{(1)}\otimes x^{(2)}\right)y^*\log \left(x^{(1)}x^{(1)*}\otimes x^{(2)}x^{(2)*}\right)\right]=\\
&\tr\left[x^{(1)}y^{(1)*}\log (x^{(1)}x^{(1)*})\right]+\tr\left[x^{(2)}y^{(2)*}\log (x^{(2)}x^{(2)*})\right]
\end{align*}
for all $y\in (x^{(1)}\otimes x^{(2)})^{\perp}$, where $y^{(1)*}\equiv\tr_2[(I\otimes x^{(2)})y^*]$ and 
$y^{(2)*}\equiv\tr_1[(x^{(1)}\otimes I)y^*]$. In the equation above we used the additivity of the logarithm function
under tensor products. Moreover, since $y\in (x^{(1)}\otimes x^{(2)})^{\perp}$, we also have
$y^{(1)}\in (x^{(1)})^{\perp}$ and $y^{(2)}\in (x^{(1)})^{\perp}$. Thus, if $x^{(1)}$ and $x^{(2)}$ are critical points, 
$x^{(1)}\otimes x^{(2)}$ is also critical~\footnote{In~\cite{AIM08,FGA} it was shown to be true for a large class of functions (not only for the von-Neumann entropy) including all the $p$-norm entropy functions.}.

In the following subsection we provide one of the main ingredients for
the local additivity of the von-Neumann entropy output of a quantum channel.

\subsection{The Subadditivity of $\Phi_{\rho}^{\pm}$}

\begin{lemma}\label{phi0}
Let $\Phi,\tilde{\Phi}:\mathbb{R}\to\mathbb{R}$ be defined as in Eq.(\ref{Phi}) and Eq.~\eqref{Phi0}, respectively. 
Then, for any $r,s\in\mathbb{R}$ the following holds:
\begin{equation}\label{rs}
\Phi(rs)\leq\tilde{\Phi}(r)+\tilde{\Phi}(s)
\end{equation}
with equality if and only if $r=s$. In the operator language of Eqs.~(\ref{gam},\ref{222}), the inequality~\eqref{rs} can be expressed as
\be
\Phi_{\rho^{A}\otimes\rho^{B}}^{\pm}\leq \tilde{\Phi}_{\rho^A\otimes I^B}+\tilde{\Phi}_{I^A\otimes\rho^B}
\;=\;\tilde{\Phi}_{\rho^A}\otimes I^B+I^A\otimes\tilde{\Phi}_{\rho^B}\;.
\ee
where two operators satisfies $O_1\leq O_2$ if and only if $\tr[y^*O_1y]\leq\tr[y^*O_2y]$ for all $y$.
\end{lemma}
\begin{proof}
We need to prove that
$$
\frac{rs+1}{rs-1}\log (r^{2}s^2)\leq \frac{r^2+1}{r^2-1}\log r^2+\frac{s^2+1}{s^2-1}\log s^2
$$
This inequality is equivalent to 
$$
\left(\frac{r^2+1}{r^2-1}-\frac{rs+1}{rs-1}\right)\log r^2+\left(\frac{s^2+1}{s^2-1}-\frac{rs+1}{rs-1}\right)\log s^2\geq 0\;
$$
which is equivalent to
\begin{equation}\label{ff}
\frac{s-r}{rs-1}\left(f(r)-f(s)\right)\geq 0\;,
\end{equation}
where
$$
f(r)\equiv \frac{r}{r^2-1}\log r^2\;.
$$
That is, we need to prove that $f(r)\geq f(s)$ if $(s-r)/(rs-1)>0$ and $f(r)\leq f(s)$ if $(s-r)/(rs-1)<0$.
From symmetry under exchange of $r$ and $s$, both cases are equivalent, and therefore without lose of generality we assume 
$(s-r)/(rs-1)>0$.This inequality is satisfied if \textbf{(a)} $s>r$ and $rs>1$ or \textbf{(b)} $s<r$ and $rs<1$.
A simple analysis of the function $f(r)$ shows that $f$ is odd, and it is monotonically increasing for $-1\leq r\leq 1$ and monotonically decreasing for $|r|>1$. Moreover, note that $f(1/r)=f(r)$.

Consider case $\textbf{(a)}$: If $s>r>1$ then $f(r)\geq f(s)$ since $f$ is monotonically 
decreasing in this region. In the same way if $-1>s>r$ then $f(r)\geq f(s)$. Another possibility in this case is that
$0<r<1<1/r<s$. But since both $r$ and $1/s$ are positive and smaller than 1, we get $f(r)\geq f(1/s)=f(s)$, where we have used the fact that $f(r)$ is monotonically increasing for $|r|\leq 1$. The last possibility in this case is that $1/r>s>-1>r$. 
For this last possibility both $s$ and $1/r$ are negative numbers bigger than $-1$ and in this region $f$ is monotonically increasing. Thus, $f(r)=f(1/r)\geq f(s)$.

Consider case $\textbf{(b)}$: First note that if $s<0<r$ then $f(s)<0<f(r)$, and if $-1<s<r<1$ then $f(r)\geq f(s)$ since $f$ is monotonically increasing in this region. Another possibility in this case is that
$s<1<r<1/s$. But since both $r$ and $1/s$ are positive and bigger than 1, we get $f(r)\geq f(1/s)=f(s)$, where we have used the fact that $f(r)$ is monotonically decreasing for $r\geq 1$. 
Finally, the last possibility in this case is that $1/r<s<-1<r$. 
For this last possibility both $s$ and $1/r$ are negative numbers smaller than $-1$ and in this region $f$ is monotonically decreasing. Thus, $f(r)=f(1/r)\geq f(s)$. 

In order to prove the equality conditions, we need to show that the expression in Eq.(\ref{ff}) equals zero if and only if
$s=r$. Before proceeding to prove that, we check the case $r=1/s$. In this case, $\Phi(rs)=\Phi(1)=2$ and $\tilde{\Phi}(s)=\tilde{\Phi}(1/r)=\tilde{\Phi}(r)$. That is, if $r=1/s$ then the equality in Eq.~\eqref{rs} holds if and only if $\tilde{\Phi}(r)=1$. As pointed out earlier,
$\tilde{\Phi}(r)=1$ if and only if $r=\pm 1$. We therefore conclude that if $r=1/s$ than the equality in Eq.~(\ref{rs}) holds if and only if
$r=s=\pm 1$. Assume now $rs\neq 1$. In this case, the expression in Eq.(\ref{ff}) equals zero if and only if $f(r)=f(s)$. However,
a simple analysis of the function $f(r)$ implies that $f(r)=f(s)$ if and only if $r=s$ or $r=1/s$. Since we assumed $rs\neq 1$, we get
that $r=s$. This completes the proof.
\end{proof}

\subsection{Proof of Theorem~\ref{main}}

We can assume without loss of generality that $n_1=m_1$, $n_2=m_2$. This can always be done by adding zero rows or columns. However, in this part of the proof we also assume
that both $x^{(1)}$ and $x^{(2)}$ are non-singular. The singular case is treated separately in section~\ref{sing}.
From the singular valued decomposition (see the argument below definition~\ref{maindef}) we can assume 
without loss of generality that $x^{(1)}=\diag\{\sqrt{p_1},\ldots,\sqrt{p_{n_1}}\} $ 
and $x^{(2)}=\diag\{\sqrt{q_1},\ldots,\sqrt{q_{n_2}}\}$, 
where $p_i$ and $q_j$ are positive and $\sum_{i=1}^{n_1}p_i=\sum_{j=1}^{n_2}q_{j}=1$.

We first assume that both $x^{(1)}$ and $x^{(2)}$ are non-degenerate local minima.
We need to show that $D_{y}^{2}E(x)>0$ for all $y\in x^\perp$, where $x\equiv x^{(1)}\otimes x^{(2)}$.
The most general
$y\in \left(x^{(1)}\otimes x^{(2)}\right)^\perp$ can be written as
\begin{equation}\label{generalform}
y=c_1x^{(1)}\otimes y^{(2)}+c_2y^{(1)}\otimes x^{(2)}+c_3y'\;,
\end{equation}
where $y^{(1)}\in (x^{(1)})^\perp$, $y^{(2)}\in (x^{(2)})^\perp$, and $y'\in\left(x^{(1)}\right)^\perp\otimes \left(x^{(2)}\right)^\perp$
are all normalized. The numbers $c_j$ can be chosen to be real because we can absorb their phases in $y^{(1)}$, $y^{(2)}$, and $y'$. They also
satisfy $c_{1}^{2}+c_{2}^{2}+c_{3}^{2}=1$, so that $y$ is normalized.

Consider first the simple case where $y=x^{(1)}\otimes y^{(2)}$. In this case,
\begin{align}
E\left(\frac{x+ty}{\sqrt{1+t^2}}\right)& =E\left(x^{(1)}\otimes \frac{x^{(2)}+ty^{(2)}}{\sqrt{1+t^2}}\right)\nonumber\\
&=E\left(x^{(1)}\right)+E\left(\frac{x^{(2)}+ty^{(2)}}{\sqrt{1+t^2}}\right)\;.
\end{align}
Since $x^{(2)}$ is a non-degenerate local minimum, we must have $D_{y}^{2}E(x)>0$.
The case $y=y^{(1)}\otimes x^{(2)}$ is similar.
 
Consider now the case in which $y\in \left(x^{(1)}\right)^\perp\otimes \left(x^{(2)}\right)^\perp$. 
Using its Schmidt decomposition, we can write it as
 \begin{equation}\label{yprime}
 y=\sum_{l}c_l y^{(1)}_{l}\otimes y^{(2)}_{l}\;,
 \end{equation}
 where 
 \begin{equation}\label{ort}
 \tr[y^{(1)}_{l}y^{(1)*}_{l'}]= \tr[y^{(2)}_{l}y^{(2)*}_{l'}]=\delta_{ll'}\;,
 \end{equation}
 and $c_l$ are real numbers such that $\sum_{l}c_{l}^{2}=1$.
 
By definition we have
\begin{align}
M_{x}(y)
=\tr\left[w^{AB}\Phi_{\rho^A\otimes\rho^B}^{+}(w^{AB})+z^{AB}\Phi_{\rho^A\otimes\rho^B}^{-}(z^{AB})\right]
\;,
\label{tendelta}
\end{align}
where $w^{AB}=(y^*+y)/2$, 
$z^{AB}=i(y^*-y)/2$, $\rho^A\equiv x^{(1)}x^{(1)*}$ and $\rho^B\equiv x^{(2)}x^{(2)*}$.

Applying lemma~\ref{phi0} both to $\Phi_{\rho^A\otimes\rho^B}^{\pm}$ gives:
\begin{align*}
M_{x}(y) & \leq\tr\left[
w^{AB}\tilde{\Phi}_{\rho^A\otimes I^{B}}(w^{AB})+w^{AB}\tilde{\Phi}_{I^{A}\otimes\rho^B}(w^{AB})+z^{AB}\tilde{\Phi}_{\rho^A\otimes I^B}(z^{AB})+z^{AB}\tilde{\Phi}_{I^A\otimes \rho^B}(z^{AB})\right]\\
& =\tr\left[
y^{*}\tilde{\Phi}_{\rho^A\otimes I^B}(y)+y^{*}\tilde{\Phi}_{I^{A}\otimes\rho^B}(y)\right]\;.
\end{align*}
where $I^A$ and $I^B$ are the identity matrices in the respective spaces, and in the last equality we have used the definitions $w^{AB}=(y^*+y)/2$ 
and $z^{AB}=i(y^*-y)/2$.  Now, but substituting~\eqref{yprime} into the above equation we get
\begin{align*}
M_{x}(y)\leq\sum_{l}c_{l}^{2}\tr\left[
y_{l}^{(1)*}\tilde{\Phi}_{\rho^A}(y_{l}^{(1)})+y_{l}^{(2)*}\tilde{\Phi}_{\rho^B}(y_{l}^{(2)})\right]\;,
\end{align*}
where we have used the orthogonality relations in Eq.~\eqref{ort}.
Combining this with Eq.~\eqref{222} gives
\begin{equation}\label{laststep}
M_{x}(y)<\sum_{l}c_{l}^{2}\left(\Gamma_{x^{(1)}}(y^{(1)}_{l})+\Gamma_{x^{(2)}}(y^{(2)}_{l})\right)=\Gamma_{x}(y)
\end{equation}
where the last equality can be verified from the orthogonality relations given in Eq.~\eqref{ort}, and the fact that
\begin{equation}\label{log}
\log xx^*=\log x^{(1)}x^{(1)*}\otimes I^B+I^A\otimes\log x^{(2)}x^{(2)*}\;.
\end{equation}
This completes the proof for $y\in\left(x^{(1)}\right)^\perp\otimes \left(x^{(2)}\right)^\perp$. 

Consider now the most general case where $y\in x^\perp$ has the form given in Eq.~(\ref{generalform}).
Denote
\begin{equation}\label{hhh}
w^{AB}=\frac{1}{2}(y^{*}+y)=c_1x^{(1)}\otimes w^{(2)}+c_2w^{(1)}\otimes x^{(2)}+c_3w'\;,
\end{equation}
where $w'=(y'{}^{*}+y')/2$ and we have used
\begin{align*}
& \frac{1}{2}\left(x^{(1)*}\otimes y^{(2)*}+x^{(1)}\otimes y^{(2)}\right)
=x^{(1)}\otimes\frac{1}{2}\left(y^{(2)*}+y^{(2)}\right)\equiv x^{(1)}\otimes w^{(2)}\\
&\frac{1}{2}\left(y^{(1)*}\otimes x^{(2)*}+y^{(1)}\otimes x^{(2)}\right)
=\frac{1}{2}\left(y^{(1)*}+y^{(1)}\right)\otimes x^{(2)}\equiv w^{(1)}\otimes x^{(2)}\;.
\end{align*}
In the above equation we used the fact that $x^{(1)}$ and $x^{(2)}$ are square diagonal matrices with their singular values on the diagonal. We would like to substitute the expression in Eq.~\eqref{hhh} for $w^{AB}$, into the expression for $M_{x}(y)$ given in Eq.~\eqref{tendelta}.
By doing that we will get expressions with several cross terms. We argue that these cross terms vanish.
To see that consider for example the cross term
$$
c_1c_3\tr\left[x^{(1)}\otimes w^{(2)}\Phi_{\rho^A\otimes\rho^B}^{+}(w^{\prime})\right]\;,
$$
and recall that $\rho^A\equiv x^{(1)}x^{(1)*}$ and $\rho^B\equiv x^{(2)}x^{(2)*}$.
Since $\Phi_{\rho^A\otimes\rho^B}^{+}$ is self-adjoint, the above expression can be written as
$$
c_1c_3\tr\left[x^{(1)}\otimes w^{(2)}\Phi_{\rho^A\otimes\rho^B}^{+}(w^{\prime})\right]=c_1c_3\tr\left[w^{\prime}\Phi_{\rho^A\otimes\rho^B}^{+}(x^{(1)}\otimes w^{(2)})\right]=
c_1c_3\tr\left[w^{\prime}\left(x^{(1)}\otimes \Phi_{\rho^B}^{+}(w^{(2)})\right)\right]
$$
where in the last equality we used the identity $\Phi_{\rho^A\otimes\rho^B}^{+}(x^{(1)}\otimes w^{(2)})
= x^{(1)}\otimes \Phi_{\rho^B}^{+}(w^{(2)})$. This identity follows from the definition of $\Phi_{\rho}^{+}$, when working with a basis in which $x^{(1)}$ is diagonal. Now, since the partial trace $\tr_1[w' (x^{(1)}\otimes B)]=0$
for all matrices $B$, we have
$$
c_1c_3\tr\left[x^{(1)}\otimes w^{(2)}\Phi_{\rho^A\otimes\rho^B}^{+}(w^{\prime})\right]=0\;.
$$
In the same way, we see that all the other cross terms vanish.
Moreover, denote
$$
z^{AB}=\frac{i}{2}(y^{*}-y)=c_1x^{(1)}\otimes z^{(2)}+c_2z^{(1)}\otimes x^{(2)}+c_3z'\;,
$$
where  $z^{(1)}$, $z^{(2)}$, and $z'$ are defined similarly to $w^{(1)}$, $w^{(2)}$, and $w'$.
Substituting this expression for $z^{AB}$ in Eq.~\eqref{tendelta} will also lead to vanishing cross terms.
To summarize, by substituting the above expressions for $z^{AB}$ and $w^{AB}$ in Eq.~\eqref{tendelta} we get
$$
M_{x}(y)=c_{1}^{2}M_{x}(x^{(1)}\otimes y^{(2)})+c_{2}^{2}M_{x}(y^{(1)}\otimes x^{(2)})
+c_{3}^{2}M_{x}(y')
$$
However, since we already proved that $x$ is a non-degenerate local minimum in the directions $x^{(1)}\otimes y^{(2)}$,
$y^{(1)}\otimes x^{(2)}$, and $y'$, we get
\begin{equation}\label{gmm}
M_{x}(y)<c_{1}^{2}\Gamma_{x}(x^{(1)}\otimes y^{(2)})+c_{2}^{2}\Gamma_{x}(y^{(1)}\otimes x^{(2)})
+c_{3}^{2}\Gamma_{x}(y')
\end{equation}
Now, note the orthogonality relations in the partial traces: $\tr_1[(x^{(1)}\otimes y^{(2)})(y')^*]=\tr_2[(x^{(1)}\otimes y^{(2)})(y')^*]=0$ and
$\tr_1[(y^{(1)}\otimes x^{(2)})(y')^*]=\tr_2[(y^{(1)}\otimes x^{(2)})(y')^*]=0$.
With these relations and from Eq.~\eqref{log} 
we get that the expression in the RHS of Eq.(\ref{gmm}) is equal to
$\Gamma_{x}(y)$. 
This completes the proof of the main part of the theorem.

To prove the second part of the theorem, assume that $x^{(1)}$ is degenerate local minimum and $x^{(2)}$ is a non-degenerate local minimum. Following the exact same lines of the proof above we get that $M_{x}(y')<\Gamma_{x}(y')$ for 
$y'\in\left(x^{(1)}\right)^\perp\otimes \left(x^{(2)}\right)^\perp$. This is clear from Eq.~(\ref{laststep}) and the one above it, 
where we use the fact that
$$
\tr\left[y_{l}^{(2)*}\tilde{\Phi}_{\rho^B}(y_{l}^{(2)})\right]<\Gamma_{x^{(2)}}(y_{l}^{(2)})
$$
since $x^{(2)}$ is a non-degenerate local minimum. Similarly, if $y=x^{(1)}\otimes y^{(2)}$ we get 
$M_{x}(y)<\Gamma_{x}(y)$. The only $y\in x^\perp$ for which it is possible to have 
$M_{x}(y)=\Gamma_{x}(y)$ is $y=y^{(1)}\otimes x^{(2)}$. However, in this case 
\begin{equation}\label{xy}
E\left(\frac{x+ty}{\sqrt{1+t^2}}\right) =
E\left(\frac{x^{(1)}+ty^{(1)}}{\sqrt{1+t^2}}\right)+E\left(x^{(2)}\right)\;,
\end{equation}
so $x$ is a local minimum in this direction as well. Hence, $x$ is a degenerate local minimum.
This completes the proof of the second part of the theorem.
 
\section{The Singular Case}\label{sing}

In the previous section, we were able to derive the first and second directional derivatives $D_{y}^{1}E(x)$ and $D_{y}^{2}E(x)$ assuming $x$ is non-singular. 
In this section we consider the case where $x$ is singular. While the expression for $D_{y}^{1}E(x)$ is the same as in the previous section, the expression for the second derivative is not the same for the singular case. In particular, in the singular case it is possible that $D_{y}^{2}E(x)\equiv\frac{d^2}{dt^2 }S(\rho(t))\big|_{t=0}$ diverge.
Nevertheless, we will see in this section that even if $x$ is singular, $E(x)$ is additive.

For simplicity of the exposition, we will consider here subspaces $\mathcal{K}\subset\mathbb{C}^{n}\otimes\mathbb{C}^{m}$, where $n=m$, since we can always embed  $\mathcal{K}$ in $\mathbb{C}^{\max{\{n,m\}}}\otimes\mathbb{C}^{\max{\{n,m\}}}$. The following theorem provides the criterion for the divergence of the second derivative.

\begin{theorem}\label{secdersin}  Let $x,y\in\mathcal{K}\subset \C^{n\times n}$, $\Tr xx^*=\tr yy^*=1$ and $\tr(xy^*)=0$.  
Change the standard orthonormal base in $\C^n$ to a new orthonormal base such that $x$ and $y$ have the forms
\begin{equation}\label{xypartition}
x=\left[\begin{array}{cc} x_{11}&0_{r,n-r}\\0_{n-r,r}&0_{n-r,n-r}\end{array}\right]\;\text{ and  }\;\;
y=\left[\begin{array}{cc} y_{11}&y_{12}\\y_{21}&y_{22}\end{array}\right], 
\end{equation}
where $r$ is the rank of $x$, $0_{i,j}$ are $i\times j$ zero matrices, and $x_{11}, y_{11}\in \C^{r\times r}$. 
Then
 \begin{equation}\label{srhotexpan}
 S(\rho(t))=f(t) -(K+tg(t))t^2\log t^2, \quad K=\tr(y_{22}y_{22}^*),
 \end{equation}
 where  $f(t),g(t)$ are analytic functions in a neighbourhood of $0$.
 Hence $D^2_yE(x)=+\infty$ if and only if $y_{22}\ne 0$.  Furthermore, if $y_{22}=0$ then either $g(t)\equiv 0$ or $g(t)=at^{2k-1}(1+O(t))$, where
 $a>0$ and $k$ is a positive integer.
 \end{theorem}
A much weaker version of the theorem above can be found in~\cite{FGA}. For the clarity of the exposition in this section,
we leave the proof of Theorem~\ref{secdersin} to appendix~\ref{appsec}.

From the theorem above it follows that w.l.o.g we can set $y_{22}=0$ since otherwise the second derivative is $+\infty$.
This will be useful when proving local additivity for the singular case. However, in the tensor product space, $y$ can be written as in Eq.(~\ref{yprime}). Hence, while we assume that the $(2,2)$ block of the bipartite state $y$ is zero, it is not immediately obvious that the $(2,2)$ blocks of the one-party states $y_{l}^{(1)}$ and $y_{l}^{(2)}$ are also zero. Nevertheless, this is indeed the case as we show now.

 \subsection{Tensor product structure in the singular case}

 Let $\mathcal{K}\subset \C^{n\times n}$ be a subspace of matrices that are partitioned as in Eq.~\eqref{xypartition}.
 We assume that $\mathcal{K}$ contains a matrix
 \begin{equation}\label{xform}
 x=\left[\begin{array}{cc} x_{11}&0\\0&0\end{array}\right], \quad \tr(x_{11}^*x_{11})=1.
 \end{equation}
 We now choose a following orthonormal base $x_1,\ldots,x_p, y_1,\ldots,y_q,z_1,\ldots,z_r,w_1,\ldots,w_s\in\mathcal{K}$.
 First, $x_1=x$.  Then
 \begin{enumerate}
 \item
 $x_1,\ldots,x_p$ is an orthonormal basis of the subspace of $\mathcal{K}$ of matrices of the form $\left[\begin{array}{cc} *&0\\0&0\end{array}\right]$.
 (It is possible that $p=1$.)
 \item
 $x_1,\ldots,x_p,y_1,\ldots,y_q$ is an orthonormal basis of the subspace of $\mathcal{K}$ of matrices of the form $\left[\begin{array}{cc} *&*\\0&0\end{array}\right]$.
 (It is possible that $q=0$.)
 \item
 $x_1,\ldots,x_p,y_1,\ldots,y_q,z_1,\ldots,z_r$ is an orthonormal basis of the subspace of $\mathcal{K}$ of matrices of the form
 $\left[\begin{array}{cc} *&*\\ *&0\end{array}\right]$.
 (It is possible that $r=0$.)
 \item
 $x_1,\ldots,x_p,y_1,\ldots,y_q,z_1,\ldots,z_r,w_1,\ldots,w_s$ is an orthonormal basis of $\mathcal{K}$.
 (It is possible that $s=0$.)
 \end{enumerate}
 We observe the following
 \begin{enumerate}
 \item The projections of $x_1,\ldots,x_p$ on the block $(1,1)$ are linearly independent.
 \item The projections of $y_1,\ldots,y_q$ on the block $(1,2)$ are linearly independent if $q\ge 1$.
 \item The projections of $z_1,\ldots,z_r$ on the block $(2,1)$ are linearly independent if $r\ge 1$.
 \item The projections of $w_1,\ldots,w_s$ on the block $(2,2)$ are linearly independent if $s\ge 1$.
 \end{enumerate}

 We now consider two subspaces $\mathcal{K}_i\subset \C^{n_i\times n_i}$ for $i=1,2$.  We consider here the most complicated case
 in which both matrices $x^{(1)}\in\mathcal{K}_1$ and $x^{(2)}\in\mathcal{K}_2$ are singular.  So we assume that each $x^{(i)}$ has the form \eqref{xform}.
For $i=1,2$ we form orthonormal bases
 \begin{equation}\label{bases}
 x_1^{(i)},\ldots,x_{p_i}^{(i)}, y_1^{(i)},\ldots,y_{q_i}^{(i)},z_1^{(i)},\ldots,z_{r_i}^{(i)},w_1^{(i)},\ldots,w_{s_i}^{(i)}\in\mathcal{K}_i
 \end{equation}
 exactly as above.  
 We now form a tensor product of $\mathcal{K}_1\otimes\mathcal{K}_2$ with respect to the partitions of $\mathcal{K}_1,\mathcal{K}_2$ as above.  
 
 Let
 \begin{equation}\label{ABpart}
 A=\left[\begin{array}{cc} A_{11}&A_{12}\\ A_{21}&A_{22}\end{array}\right]\in\mathcal{K}_1,\;
 B=\left[\begin{array}{cc} B_{11}&B_{12}\\ B_{21}&B_{22}\end{array}\right]\in\mathcal{K}_2.
 \end{equation}
 We then agree that the partition in $\mathcal{K}_1\otimes\mathcal{K}_2$ is of the form as the following partition of $A\otimes B$:
 \begin{equation}\label{AtensBpart}
 A\otimes B=\left[\begin{array}{cccc} A_{11}\otimes B_{11}&A_{11}\otimes B_{12}& A_{12}\otimes B_{11}&A_{12}\otimes B_{12}\\
 A_{11}\otimes B_{21}&A_{11}\otimes B_{22}& A_{12}\otimes B_{21}&A_{12}\otimes B_{22}\\
 A_{21}\otimes B_{11}&A_{21}\otimes B_{12}& A_{22}\otimes B_{11}&A_{22}\otimes B_{12}\\
 A_{21}\otimes B_{21}&A_{21}\otimes B_{22}& A_{22}\otimes B_{21}&A_{22}\otimes B_{22}
 \end{array}\right].
 \end{equation}

 \begin{lemma}\label{tensorlem1}  Let $\mathcal{K}_1,\mathcal{K}_2$ be two subspaces in $\C^{n_{1}\times n_{1}}$ and $\C^{n_{2}\times n_{2}}$, respectively. Let $C=[C_{ij}]_{i,j=1}^4\in\mathcal{K}_1\otimes\mathcal{K}_2$ be partitioned as
 in \eqref{AtensBpart}.  Suppose that $C\ne 0$ and $C_{ij}=0$ for $i,j\ge 2$.  Write $C$ as a linear combination of the tensor products of the bases of $\mathcal{K}_1$ and $\mathcal{K}_2$, chosen as in~\eqref{bases}.  Then each term in this linear combination of $C$ is of the form $\alpha f\otimes g$,
 where $\alpha\in\mathbb{C}$, $f\in\mathcal{K}_1$, $g\in\mathcal{K}_2$, and both $f$ and $g$ have the form 
 $\left[\begin{array}{cc} *&*\\ *&0\end{array}\right]$.
 \end{lemma}
 
 \begin{remark}
 It is also possible to show that at least one of the matrices $f$ and $g$ must have the form $\left[\begin{array}{cc} *&*\\ 0&0\end{array}\right]$  or $\left[\begin{array}{cc} *&0\\ *&0\end{array}\right]$. However, we will not be using it here.
 \end{remark}
  
 \begin{proof}
 Suppose the expansion of $C$ contains a term of the form $w^{(1)}_i\otimes w^{(2)}_j$.  Look at the block $(4,4)$.  The contribution of the
 expansion of C to this block only comes from the tensor products projections of $w^{(1)}_i$ and $w^{(2)}_j$ on the block $(2,2)$.
 Since all these projections are linearly independent we must have that $C_{44}\ne 0$ contrary to our assumption.

 Assume now that the expansion of $C$ contains $w_{i}^{(1)}\otimes z_j^{(2)}$.   Since  the expansion of $C$ does not have terms  $w^{(1)}_i\otimes w^{(2)}_j$, the contribution to the block
 $C_{43}$ comes only from the projection of $w_{i}^{(1)}$ on the block $(2,2)$ and the projection of $z_j^{(2)}$ on the block $(2,1)$.
 Again as all these projections are linearly independent we deduce that $C_{43}\ne 0$, contrary to our assumptions.
 
Similarly, there are no terms in the expansion of $C$ of the form $w_i^{(1)}\otimes y_j^{(2)}$,  since $C_{34}=0$,  and there are no terms in the expansion of $C$ of the form $w_i^{(1)}\otimes x_j^{(2)}$ since $C_{33}=0$.
That is, we have shown that the matrices $w_{i}^{(1)}$ do not appear in the expansion of $C$. 
In exactly the same way, there are no terms in the expansion  of $C$ of the form $z_{i}^{(1)}\otimes w_j^{(2)}$,
$y_{i}^{(1)}\otimes w_j^{(2)}$, and $x_{i}^{(1)}\otimes w_j^{(2)}$, since $C_{42}=0$, $C_{24}=0$, and
$C_{22}=0$, respectively. This completes the proof.
 \end{proof}

\subsection{Local additivity in the singular case}

In this subsection we prove Theorem~\ref{main} for the case in which $x^{(1)}$ and $x^{(2)}$ are \emph{singular} local minima of $\mathcal{K}^{(1)}$ and $\mathcal{K}^{(2)}$, respectively.
We therefore choose bases such that $x^{(1)}$ and $x^{(2)}$ are of the form given in Eq.~\eqref{xypartition}, and denote by $r_1$ and $r_2$ their respective ranks. 

Assume first that both $x^{(1)}$ and $x^{(2)}$ are \emph{non-degenerate} local minima.
We need to show that $D_{y}^{2}E(x)>0$ for all $y\in x^\perp$, where 
$x\equiv x^{(1)}\otimes x^{(2)}$. Note that the partition of $x=[x_{ij}]_{i,j=1}^{4}$ as in Eq.~(\ref{AtensBpart})
gives $x_{ij}=0$ for all $i,j=1,2,3,4$ except for $x_{11}=x^{(1)}_{11}\otimes x^{(2)}_{11}$.

The most general
$y\in \left(x^{(1)}\otimes x^{(2)}\right)^\perp$ can be written as in Eq.~(\ref{generalform}), where $y'$ is of the form given in Eq.~(\ref{yprime}). Consider now
the partition of $y=[y_{ij}]_{i=j=1}^4$ as in Eq.(\ref{AtensBpart}).
From Theorem~\ref{secdersin} we know that $D_{y}^{2}E(x)=+\infty$ unless
$y_{ij}=0$ for all $i,j=2,3,4$. We therefore assume now that $y_{ij}=0$ for all $i,j=2,3,4$.
In this case, Lemma~\ref{tensorlem1} implies that all the matrices $y_{l}^{(1)}$ and $y_{l}^{(2)}$ 
in Eq.~(\ref{yprime}) have the form 
$$
\left[\begin{array}{cc} *&*\\ *&0\end{array}\right].
$$
That is, their $(2,2)$ block is zero. For this reason, 
we replace each subspace $\mathcal{K}^{(i)}\subset\C^{n_i\times n_i}$ $(i=1,2)$ with a smaller subspace
$\mathcal{U}^{(i)}\subset \mathcal{K}^{(i)}$ such that each matrix in the basis of $\mathcal{U}^{(i)}$
has zeros on the (2,2) block.  It is left to prove that $x\equiv x_1\otimes x_2$ is local
minimum in $\mathcal{U}^{(1)}\otimes \mathcal{U}^{(2)}$.

Consider the new subspace $\mathcal{U}^{(i)}_{\epsilon}$, for $\epsilon>0$,
where in the orthonormal basis of $\mathcal{U}^{(i)}$, we change only the first
matrix $x^{(i)}$, i.e. the local minimum matrix, with the normalized diagonal matrix
$$
x^{(i)}_{\epsilon}\equiv\frac{1}{\sqrt{1+(n_i-r_i)\epsilon^2}}
\left[\begin{array}{cc} x_{11}^{(i)}&0_{r_i,n_i-r_i}\\ 0_{n_i-r_i,r_i}& \epsilon I_{n_i-r_i}\end{array}\right]\;\;\;i=1,2
$$
where $0_{i,j}$ are $i\times j$ zero matrices and $I_{n_i-r_i}$ are $(n_i-r_i)\times(n_i-r_i)$ identity matrices.

\begin{lemma}\label{yxeps}
Assume $x^{(i)}$ is a non-degenerate local minimum in $\mathcal{U}^{(i)}$, then
$x^{(i)}_{\epsilon}$ is a non-degenerate local minimum in $\mathcal{U}^{(i)}_{\epsilon}$. Moreover, there exists $\delta>0$ and $\epsilon_0>0$ such that if $\epsilon<\epsilon_0$ then $D_{y^{(i)}}^{2}E(x^{(i)}_{\epsilon})>\delta$
for all $y^{(i)}\in\left(x^{(i)}_{\epsilon}\right)^\perp$.
\end{lemma}

\begin{proof}
For simplicity of the exposition we remove the superscript $(i)$ from $x^{(i)}$ and denote $d\equiv n-r$. That is, consider
$$
x=
\left[\begin{array}{cc} x_{11}&0_{r,d}\\ 0_{d,r}& 0_{d,d}\end{array}\right]\;\text{and}\;
x_{\epsilon}\equiv\frac{1}{\sqrt{1+d\epsilon^2}}
\left[\begin{array}{cc} x_{11}&0_{r,d}\\ 0_{d,r}& \epsilon I_{d}\end{array}\right]\;.
$$
We need to show that if $x$ is a non-degenerate local minimum in $\mathcal{U}$ then $x_\epsilon$ is a non-degenerate local minimum in $\mathcal{U}_\epsilon$ for small enough $\epsilon$.

First, we need to show that $x_\epsilon$ remains critical. Indeed, since the condition~\eqref{critical} for criticality
is satisfied for $x$, it is also satisfied for $x_\epsilon$. This is because $x_\epsilon$ is a diagonal matrix and all $y\in x_{\epsilon}^{\perp}\subset\mathcal{U}$ is of the form 
$
\left[\begin{array}{cc} *&*\\ *&0_{d,d}\end{array}\right]\;.
$

Second, we need to show that $D_{y}^{2}E(x_\epsilon)>\delta$. In Appendix~\ref{appb} we show
that $D_{y}^{2}E(x_\epsilon)$ does \emph{not} diverge in the limit $\epsilon\to 0$ (assuming $y_{jk}=0$ when both $j>r$ and $k>r$). Now, since we assume $D_{y}^{2}E(x)>0$
for all $y\in x^\perp$, we
can also assume that there exist $\delta'>0$ such that $D_{y}^{2}E(x)>\delta'$ for all $y\in x^\perp$. This is true because the set of all normalized matrices in $x^\perp$ is compact. Hence, from the nice behaviour of  $D_{y}^{2}E(x_\epsilon)$ in the limit
$\epsilon\to 0$ (see Appendix~\ref{appb}), we get that
for small enough $\epsilon$ there exists $\delta>0$
such that $D_{y}^{2}E(x_\epsilon)>\delta$ for all $y\in (x_\epsilon)^\perp$. This completes the proof of the lemma.
\end{proof}

We now apply Theorem~\ref{main} to the non-singular case of $x_\epsilon\equiv x^{(1)}_\epsilon\otimes  x^{(2)}_\epsilon$.
From Lemma~\ref{yxeps}, the second derivatives $D_{y^{(i)}}^{2}E(x^{(i)}_{\epsilon})>\delta$
for all $y^{(i)}\in\left(x^{(i)}\right)^\perp$ and $i=1,2$. Thus, we get that $D_{y}^{2}E(x_\epsilon)>2\delta$ for all $y\in (x_\epsilon)^\perp$. We obtain it by following precisely the same steps of the proof of Theorem~\ref{main} (in the non-singular case). Letting $\epsilon\to 0$ we deduce that in the direction
of $y$ the second derivative at $x^{(1)}\otimes x^{(2)}$ is strictly positive (greater or equal to $2\delta$). This complete the proof of the main part of theorem~\ref{main} for the singular case.

To proof the second part of the theorem, we assume now that $x^{(1)}$ is degenerate local minimum and $x^{(2)}$ is non-degenerate local minimum. In this case we only have $D_{y^{(2)}}^{2}E(x^{(2)}_{\epsilon})>\delta$. Nevertheless,
in Appendix~\ref{appb} we show that $D_{y^{(1)}}^{2}E(x^{(1)}_{\epsilon})-D_{y^{(1)}}^{2}E(x^{(1)})$ is of order $\epsilon\log\epsilon^2$. Therefore, since $D_{y^{(1)}}^{2}E(x^{(1)})\geq 0$, it follows that we can choose $\epsilon$ small enough such that $D_{y^{(1)}}^{2}E(x^{(1)}_{\epsilon})>-\delta/2$.

As pointed out in the proof of the non-singular case of theorem~\ref{main}, the only $y\in x^\perp$ (recall $x\equiv x^{(1)}\otimes x^{(2)}$) for which it is possible to have 
$D_{y}^{2}E(x)=0$ is $y=y^{(1)}\otimes x^{(2)}$. However, the equality in Eq.~\eqref{xy} implies
that $x$ is a local minimum in this direction and this is also true even if $x^{(i)}$ are singular. We will therefore assume now that $y$ is not of the form $y^{(1)}\otimes x^{(2)}$.
By following precisely the same steps of the proof of Theorem~\ref{main} (in the non-singular case) we get that
for all other $y\in x^\perp$ we have $D_{y}^{2}E(x_\epsilon)>\delta-\delta/2=\delta/2$. We therefore get  
$D_{y}^{2}E(x)>0$ in the limit $\epsilon\to 0$. This completes the proof of the second part of theorem~\ref{main}.

\section{Discussion}\label{conc}

We have shown that the minimum entropy output of a quantum channel is locally additive 
(assuming at least one of the two local minima is non-degenerate). 
Our proof consists of two key ingredients.  The first one is
the use of the divided difference approach, which enabled us to calculate directional derivatives explicitly, and
the second one is the explicit use of the complex structure. In the appendix B of~\cite{FGA} we show that there exists counterexamples for local additivity over the real numbers. These counterexamples precludes the existence of 
a more straightforward differentiation argument than the complex structure based argument given here.

The fact that the minimum entropy output is not globally additive makes local additivity of even greater interest
to quantum information theorists. It suggests that it is some global feature, of the quantum channels involved, that corresponds to cases of non-additivity of the minimum entropy output. 
Perhaps one way to improve our understanding in this direction is to study properties of generic channels.
In particular, it seems quite possible to us that for generic channels (or generic subspaces) the entropy output have a 
\emph{finite} number of isolated non-degenerate critical points.

\emph{Acknowledgments:---}
We acknowledge many fruitful discussions with A. Roy and J. Yard in the earlier stages of this work.
GG research is supported by NSERC. The authors acknowledge support from PIMS CRG MQI, MITACS, and iCore for Shmuel Friedland's visits to IQIS in Calgary.

\begin{appendix}
\section{\label{appsec}Proof of Theorem~\ref{secdersin}}

 \proof  Let $\lambda_1(t)\ge\ldots\ge \lambda_n(t)\ge 0$, for $t>0$, be the eigenvalues of $\rho(t)$.  Rellich's theorem yields that
 each $\lambda_i(t)$ is analytic in $t$ in a neighbourhood of $t=0$.  So $\lambda_i(0)=\lambda_i(\rho)>0$ for $i=1,\ldots,r$ and $\lambda_i(0)=0$
 for $i=r+1,\ldots,n$. Since each $\lambda_i(t)\ge 0$ it follows that the Taylor expansion of each  $\lambda_i(t)\not \equiv 0$, for $i>r$,  
 must start with $t$ to a positive even power times a positive constant.   I.e. 
 $\lambda_i(t)=\lambda_{i, 2n_i}t^{2n_i}(1+O(t))$, where $\lambda_{i,2n_i}>0$ and $n_i$ is a positive integer for $i>r$.  This shows that $S(\rho(t))=-\sum_{i=1}^n \lambda_i(t)\log\lambda_i(t)$ must be of the form \eqref{srhotexpan}.  Furthermore, $K=0$ if and only if $n_i\ge 2$ for all $i>r$.  So if $K=0$ and not all $\lambda_i(t)$ are identically zero for $i>r$, then $k=\min\{n_i-1, \lambda_{i, 2n_i}>0\}$.

 It is left to show that $K=\tr(y_{22}y_{22}^*)$.
 Let $X=\left[\begin{array}{cc}0&x\\x^*&0\end{array}\right], Y=\left[\begin{array}{cc}0&y\\y^*&0\end{array}\right]$.
 Recall that the pencil $X+tY$ has $n$ nonnegaive and $n$ nonpositive eigenvalues
 \[\sigma_1(t)\ge\ldots\ge\sigma_n(t)\ge 0\ge-\sigma_n(t)\ge\ldots\ge-\sigma_1(t).\]
 The singular values of $x+ty$ are the $n$ nonnegative eigenvalues of $X+tY$.  Hence, the eigenvalues of $\rho(t)$
 are $\frac{\sigma_i(t)^2}{1+t^2}$ for $i=1,\ldots,n$.  Let $\sigma_i(t)=\sigma_{i,1}t +O(t^2)$ for $t>0$ and $i>r$.
 Hence the coefficient of $t^2$ in the $i$-th eigenvalue of $\rho(t)$, for $i>r$, is $\sigma_{i,1}^2$.
 Thus $K=\sum_{i=r+1}^n \sigma_{i,1}^2$.\\

 Let $P\in \C^{2n\times 2n}$ be the orthogonal projection on the zero eigenspace of $X$.  Then $PYP((I-P)\C^{2n})=\0$.
 The other possible nonzero eigenvalues of $PYP$ are $\sigma_{r+1,1}\ge\ldots\ge\sigma_{n,1}\ge 0\ge -\sigma_{n,1}\ge\ldots\ge -\sigma_{r+1,1}$,
 which are the eigenvalues of the restriction of $PYP$ on the kernel of $\rho$ \cite{Fri78,Kato} or \cite[\S3.8]{Fri10}.  The restriction of $PYP$ to the kernel
 of $X$ is $\left[\begin{array}{cc}0&y_{22}\\y_{22}^*&0\end{array}\right]$, obtained by deleting the corresponding rows and the columns in $Y$.
 Hence
 \[2K=2\sum_{i=r+1}^n \sigma_{i,1}^2 =\tr((PYP)^2)=\tr(y_{22}y_{22}^*+y_{22}^*y_{22}).\]
 This completes the proof.\qed

\section{Formula for the second derivative in the singular case\label{appb}}

\begin{proposition}
Let 
\begin{equation}\label{xypartition2}
x_\epsilon=\left[\begin{array}{cc} x_{11}&0_{r,n-r}\\0_{n-r,r}&\epsilon I_{n-r,n-r}\end{array}\right]\;\text{ and  }\;\;
y=\left[\begin{array}{cc} y_{11}&y_{12}\\y_{21}&0_{n-r,n-r}\end{array}\right], 
\end{equation}
where $0_{i,j}$ are $i\times j$ zero matrices, and $x_{11}, y_{11}\in \C^{r\times r}$, $y_{12}\in \C^{r\times n}$,
$y_{21}\in \C^{n\times r}$. We also assume that $x_{11}=\diag\{\sqrt{p_1},...,\sqrt{p_r}\}$ is non singular. Then, the limit of $D_{y}^{2}E(x_\epsilon)$ when $\epsilon$ goes to zero exists and equals to
\be\label{mainappen}
\lim_{\epsilon\to 0}D_{y}^{2}E(x_\epsilon)
=D_{y_{11}}^{2}E(x_{11})-2\tr\left[\left(y_{12}y_{12}^{*}+y_{21}^{*}y_{21}\right)\log\rho_{11}\right]\;,
\ee
where $\rho_{11}\equiv x_{11}x_{11}^{*}$.
\end{proposition}
\begin{remark}
The contribution of the normalization factor of $x_\epsilon$ is of order $O(\epsilon^2)$ and therefore ignored here.
\end{remark}
\proof The proof is based on a straightforward calculation.
The expression for the second derivative given in Eq.(\ref{se}), can be written as:
$$
D_{y}^{2}E(x_\epsilon)  \equiv\frac{d^2}{dt^2}S(\rho_\epsilon(t))\Big|_{t=0}
=-2\left(S(\rho_\epsilon)+\sum_{j=1}^{n}\sum_{k=1}^{n}G_{jk}\right)\;,
$$
where $\rho_\epsilon\equiv x_\epsilon x^{*}_{\epsilon}$, $\gamma_0\equiv x_\epsilon y^*+yx_{\epsilon}^{*}$
\be\label{321}
G_{jk}\equiv|y_{jk}|^2\log p_j+ \frac{\log p_{j}-\log p_{k}}{2(p_j-p_k)}\left|(\gamma_0)_{jk}\right|^2\;.
\ee 

If both $j,k$ are smaller or equal to $r$, then clearly those $G_{jk}$ terms contribute to $D_{y_{11}}^{2}E(x_{11})$.
Now, if both $j>r$ and $k>r$ then $y_{jk}=0$ and we have $G_{jk}=0$. Hence, we get
\be\label{gres}
D_{y}^{2}E(x_\epsilon) =D_{y_{11}}^{2}E(x_{11})-2\sum_{j=r+1}^{n}\sum_{k=1}^{r}\left(G_{jk}+G_{kj}\right)
\ee
We therefore focus now on the expressions for $G_{jk}$ and $G_{kj}$ in the case $j>r$ and $k\leq r$.

Writing $x_\epsilon=\diag\{\sqrt{p_1},...,\sqrt{p_n}\}$ with $p_j=\epsilon^2$ for $j>r$ we have
\begin{align*}
(\gamma_{0})_{jk}&=\sqrt{p_j}\bar{y}_{kj}+\sqrt{p_k}y_{jk}=\sqrt{p_k}y_{jk}+O(\epsilon)\\
(\gamma_{0})_{kj}&=\sqrt{p_k}\bar{y}_{jk}+\sqrt{p_j}y_{kj}=\sqrt{p_k}\bar{y}_{jk}+O(\epsilon)\;,
\end{align*}
where the last equality was obtained by setting $p_j=\epsilon^2$. We therefore have 
$
\left|(\gamma_0)_{jk}\right|^2=\left|(\gamma_0)_{kj}\right|^2
$
up to $O(\epsilon)$. From the expressions above we get for  $j>r$ and $k\leq r$ the following formulas:
\begin{align*}
G_{jk}&=|y_{jk}|^2\log p_j+ \frac{\log p_{j}-\log p_{k}}{2(p_j-p_k)}\left(p_k|y_{jk}|^2+O(\epsilon)\right)\\
G_{kj}&=|y_{kj}|^2\log p_k+ \frac{\log p_{k}-\log p_{j}}{2(p_k-p_j)}\left(p_k|y_{jk}|^2+O(\epsilon)\right)
\end{align*}
Since $p_j=\epsilon^2$ and $p_k>0$, we have
\begin{align*}
G_{jk}&=|y_{jk}|^2\log \epsilon^2+ \frac{\log \epsilon^2-\log p_{k}}{2(\epsilon^2-p_k)}\left(p_k|y_{jk}|^2+O(\epsilon)\right)=\frac{1}{2}|y_{jk}|^2\log\epsilon^2+\frac{1}{2}|y_{jk}|^2\log p_k+O(\epsilon\log\epsilon)\\
G_{kj}&=|y_{kj}|^2\log p_k+ \frac{\log p_{k}-\log\epsilon^2}{2(p_k-\epsilon^2)}\left(p_k|y_{jk}|^2+O(\epsilon)\right)
=|y_{kj}|^2\log p_k+\frac{1}{2}|y_{jk}|^2\log p_k-\frac{1}{2}|y_{jk}|^2\log\epsilon^2+O(\epsilon\log\epsilon)
\end{align*}
Hence,
$$
G_{jk}+G_{kj}=|y_{jk}|^2\log p_k+|y_{kj}|^2\log p_k+O(\epsilon\log\epsilon)\;.
$$
By substituting this expression into Eq.(\ref{gres}) we get~(\ref{mainappen}). This completes the proof.\qed

\end{appendix}

\end{document}